\documentclass[letterpaper, 10 pt, conference]{ieeeconf}  

\IEEEoverridecommandlockouts                              

\usepackage{latexsym,amsfonts}
\usepackage{paralist}
\usepackage{dsfont}
\usepackage{booktabs}
\usepackage{bm}
\usepackage{graphicx}
\usepackage{subfigure}
\usepackage{eso-pic}
\usepackage{mathtools}
\usepackage{epstopdf}
\usepackage{tikz}
\DeclareSymbolFont{symbolsC}{U}{pxsyc}{m}{n}
\SetSymbolFont{symbolsC}{bold}{U}{pxsyc}{bx}{n}
\DeclareFontSubstitution{U}{pxsyc}{m}{n}
\DeclareMathSymbol{\medcirc}{\mathbin}{symbolsC}{7}
\usetikzlibrary{arrows,calc,decorations.markings,math,arrows.meta}

\usepackage{algorithm}
\usepackage[noend]{algpseudocode}
\algnewcommand\algorithmicinput{\textbf{Input:}}
\algnewcommand\Input{\item[\algorithmicinput]}
\algnewcommand\algorithmicoutput{\textbf{Output:}}
\algnewcommand\Output{\item[\algorithmicoutput]}
\algnewcommand{\Initialize}[1]{%
	\State \textbf{Initialize:}
	\Statex \hspace*{\algorithmicindent}\parbox[t]{.8\linewidth}{\raggedright #1}
}

\newtheorem{theorem}{Theorem}[section]
\newtheorem{lemma}[theorem]{Lemma}
\newtheorem{problem}[theorem]{Problem}
\newtheorem{proposition}[theorem]{Proposition}
\newtheorem{remark}[theorem]{Remark}
\newtheorem{assumption}[theorem]{Assumption}

\newcommand{\R}{{\mathbb{R}}}

\newcommand{\ra}{\rightarrow}

\title{\LARGE \bf
Control Barrier Functions for Unknown Nonlinear Systems\\ using Gaussian Processes*
}

\author{Pushpak Jagtap$^{1}$, George J. Pappas$^{2}$, and Majid Zamani$^{3,4}$
\thanks{*This work was supported in part by the H2020 ERC Starting Grant AutoCPS (grant agreement number 804639), NSF grant CNS-2039062, AFOSR grant FA9550-19-1- 0265 (Assured Autonomy in Contested Environments), the German Research Foundation (DFG) through the grant ZA 873/1-1, and the TUM International Graduate School of Science and Engineering (IGSSE). }
\thanks{$^{1}$P. Jagtap is with the Electrical and Computer Engineering Department, Technical University of Munich, Germany. {\tt\small pushpak.jagtap@tum.de}}%
\thanks{$^{2}$G. J. Pappas is with the Department of Electrical and Systems Engineering, University of Pennsylvania, USA. {\tt\small pappasg@seas.upenn.edu}}%
\thanks{$^{3}$M. Zamani is with the Computer Science Department, University of Colorado Boulder, USA. {\tt\small majid.zamani@colorado.edu}}
\thanks{$^{4}$M. Zamani is with the Computer Science Department, LMU Munich, Germany.}
}

\begin{document}

\maketitle
\thispagestyle{empty}
\pagestyle{empty}

\begin{abstract}

 This paper focuses on the controller synthesis for unknown, nonlinear systems while ensuring safety constraints. Our approach consists of two steps, a learning step that uses Gaussian processes and a controller synthesis step that is based on control barrier functions. In the learning step, we use a data-driven approach utilizing Gaussian processes to learn the unknown control affine nonlinear dynamics together with a statistical bound on the accuracy of the learned model. In the second controller synthesis steps, we develop a systematic approach to compute control barrier functions that explicitly take into consideration the uncertainty of the learned model. The control barrier function not only results in a safe controller by construction but also provides a rigorous lower bound on the probability of satisfaction of the safety specification. Finally, we illustrate the effectiveness of the proposed results by synthesizing a safety controller for a jet engine example.

\end{abstract}

\section{INTRODUCTION}

The synthesis of controllers enforcing safety in safety-critical applications has gained significant attention in the past few years.  The conventional techniques for synthesizing such controllers require a precise mathematical model of the system  \cite{Tabuada.2009,belta2017formal}.  However, there are many control applications where the precise model description can not be derived or is not available analytically. In such cases, sensor or simulation data coupled with data-driven approaches from machine learning can be used to identify unmodeled dynamics with high precision and complement the mathematical analysis from control theory.

Recently, Gaussian processes (GPs) have emerged as a data-driven approach providing a non-parametric probabilistic modeling framework, to design controllers for unknown dynamical systems  \cite{kocijan2016modelling}.
The results include utilizing GPs for providing model predictive control scheme  \cite{kocijan2004gaussian,koller2018learning,hewing2019cautious}, adaptive control  \cite{chowdhary2014bayesian}, tracking control of Euler-Lagrange systems  \cite{beckers2019stable}, backstepping control  \cite{capone2019backstepping}, control Lyapunov approaches  \cite{8368325}, feedback linearization schemes  \cite{umlauft2017feedback}, safe optimization of controller \cite{berkenkamp2016safe}, and policy through reinforcement learning \cite{akametalu2014reachability}  for partially or fully unknown dynamics. 
However, all these existing works mainly focus on conventional stability or tracking objectives and not safety or complex logic specifications.

On the other hand, approaches based on control barrier functions have shown a great potential in solving controller synthesis problems ensuring safety (see \cite[and references therein]{ames2016control,ames2019control}) and more complex logic specifications  \cite{jagtap2019formal,yang2019continuous,li2019temporal}. Unfortunately, there are very few results available in the literature utilizing notions of control barrier functions for unknown dynamical systems. Assuming a prior knowledge of control barrier functions, the results in \cite{wang2018safe} and \cite{cheng2019end} provide the safe online learning of the Gaussian process model and safe learning of the reinforcement learning policy, respectively. However, these results lack a formal guarantee on the obtained control strategies and are not suitable for the system with bounded control input set. Motivated by the above results and their limitations, this work proposes the synthesis of a controller for unknown nonlinear control systems ensuring safety specifications with some probabilistic guarantee by utilizing both Gaussian processes and control barrier functions.  

To the best of our knowledge, this paper is the first to combine notions of control barrier functions and Gaussian process models to synthesize controllers enforcing safety over unknown nonlinear control affine systems. 
In order to solve this controller synthesis problem, first, we learn the Gaussian process model from the noisy measurements along with the probabilistic guarantee on the model accuracy. 
Second, we provide a systematic approach to solve safety synthesis problem by computing corresponding control barrier functions and associated controllers while considering learned GP along with its corresponding confidence. 
Finally, we demonstrate the effectiveness of the proposed results on a jet-engine example.

The remainder of this paper is structured as follows. In Section~\ref{prelim}, we introduce a class of systems and assumptions required in this work. Then, we formally define the problem considered in this paper. Section \ref{gpm} discusses the Gaussian process model and the probabilistic guaranty on the modeling accuracy. We discuss in Section~\ref{sec:cbf} a notion of control barrier functions and results for the computation of lower bounds on probabilities of safety specifications. We also provide a systematic approach for the computation of control barrier functions utilizing counter-example guided inductive synthesis approach. 
Section~\ref{sec:case} demonstrates the effectiveness of the proposed results on a jet-engine example. Finally, Section~\ref{conclusion} concludes the paper.

\section{Preliminaries and Problem Definition}\label{prelim}
\subsection*{Notations}
We denote the set of real, positive real, nonnegative real, and positive integer numbers by $\mathbb{R}$, $\mathbb{R}^+$, $\mathbb{R}_0^+ $, and $\mathbb{N}$, respectively. 
We use $\mathbb{R}^n$ to denote an $n$-dimensional Euclidean space and $\mathbb{R}^{n \times m}$ to denote a space of real matrices with $n$ rows and $m$ columns. 
We use $\bm I_n$ and $0_n$ to represent the identity matrix in $\R^{n\times n}$ and the zero vector in $\R^n$, respectively.  
We denote by $\mathcal{N}(\mu,C)$ the multivariate normal distribution with mean vector $\mu\in\R^n$ and covariance matrix $C\in\R^{n\times n}$.
For events $A_1,\ldots A_n$, we use notation $\bigcap_{i=1}^n A_i$ to represent inner product of events.
The reproducing kernel Hilbert space (RKHS) is a Hilbert space of square integrable functions that includes functions of the form $l(x)=\sum_i \alpha_i k(x,x_i)$, where $\alpha_i\in\R$, $x,x_i\in X\subset\R^n$, and $k: X\times X\ra\R^+_0$ is a symmetric positive definite function referred to as kernel. 
{The corresponding induced RKHS norm is denoted by $\|l\|_k$. 
For a detailed discussion on RKHS and RKHS norm, we kindly refer interested readers to \cite{paulsen2016introduction}.} \\

We consider nonlinear, continuous-time control affine systems $\mathcal S$ defined as 
\begin{align}\label{sys}
\dot{x}=f(x)+g(x)u,
\end{align}
where $x\in X\subset \mathbb{R}^n$ is the state vector in compact set $X$ and $u\in U\subseteq\mathbb{R}^m$ is the control input. 
We use $\phi(x_0, \upsilon, t)$ to denote the value of trajectory of the system $\mathcal S$ starting form initial state $x_0$ under the input signal $\upsilon:\R_0^+\rightarrow U$ at time $t\in\R_0^+$.
We use $f_j$ to represent the $j$th component of the vector function $f$, where $j\in\{1,\ldots,n\}$. In order to formulate the problem, we consider some assumptions as follow: 
\begin{assumption}\label{A1}
 In dynamics \eqref{sys}, the function $f:X\rightarrow \mathbb{R}^n$ is unknown and function $g:X\rightarrow \mathbb{R}^{n\times m}$ is known.	
\end{assumption} 
We also assume that the unknown function $f$ has low complexity, as measured under the reproducing kernel Hilbert space (RKHS) norm  \cite{paulsen2016introduction} as
described below.
\begin{assumption}\label{A2}
	The function $f$ in \eqref{sys} has bounded RKHS norm with respect to known kernel $k$, that is $\|f_j\|_k\leq\infty$ for all $j\in\{1,\ldots,n\}$.
\end{assumption}  
Note that, for most of the kernels used in practice, an RKHS is dense in the space of continuous functions restricted to a compact domain $X$. Thus, they can uniformly approximate any continuous function on a compact set $X$  \cite{seeger2008information}.
In addition, we require the following assumption on the availability of a training data set.
\begin{assumption}\label{A3}
	We have access to measurements $x\in X$ and $y=f(x)+w$, where $w \sim \mathcal{N}(0_n,\rho_f^2\bm I_n)$ is an additive noise.
\end{assumption}
Note that from a practical point of view, the measurements $f(x)$ (i.e., the derivative) can be obtained approximately using state measurements by running the system \eqref{sys} for sufficiently small sampling time from different initial conditions $x$ and with input signal $\upsilon\equiv 0$. To accommodate the approximation uncertainties, we consider the measurement noise $w$.\\
The main controller synthesis problem in this work is formally defined next.
\begin{problem}\label{prob1}
	Consider a system $\mathcal{S}$ in \eqref{sys} satisfying Assumptions \ref{A1}, \ref{A2}, and \ref{A3}, an initial set $X_0\subseteq X$, and unsafe set $X_1\subseteq X$. 
	We aim at computing a controller $($if existing$)$ that provides a $($possibly tight$)$ data-dependent lower bound on the probability on the trajectory  of $\mathcal{S}$ starting from any $x_0\in X_0$ avoiding unsafe set $X_1$.
\end{problem}
Finding a solution to Problem \ref{prob1} (if existing) is difficult in general. In this paper, we provide a method that is sound in solving this problem (i.e., if the method provides a solution to the problem, then we can formally conclude that it is actually a solution. However, if the method reports failure, then the solution to the problem may or may not exist).  To find a control policy $\upsilon$, we first model the unknown dynamics using Gaussian processes and then develop a notion of control barrier function that is robust to the data-dependent modeling uncertainty with some confidence.
\section{Gaussian Process Model}\label{gpm}
Gaussian processes (GPs) are a non-parametric regression method, where the goal is to find an approximation of a nonlinear map $f:X\ra\R^n$. Since $f$ is $n$-dimensional, each component $f_j$ is approximated with a Gaussian process $\hat f_j(x)\sim\mathcal{GP}(\mathsf{m}_j(x),k_j(x,x'))$, $j\in\{1,\ldots,n\}$, where $\mathsf{m}_j: X\ra\R$ is a mean function and $k_j:X\times X\ra\R$ is a covariance function (a.k.a., kernel) which measures similarity between any two states $x,x'\in X$. In general, any real-valued function can be used for $\mathsf{m}_j$ (it is common practice to set $\mathsf{m}_j(x)=0$, $\forall x\in X$ and $\forall j\in\{1,\ldots,n\}$) and the choice of kernel function is problem dependent, the most commonly used kernels include the linear, squared-exponential, and Mat\`ern kernels  \cite{williams2006gaussian}. The approximation of $f$ with $n$ independent GPs is given as 
\begin{align}
\hat f(x)=\left\{\begin{matrix}
\hat f_1(x)\sim\mathcal{GP}(0,k_1(x,x')),\\ 
\vdots \\ 
\hat f_n(x)\sim\mathcal{GP}(0,k_n(x,x')).
\end{matrix}\right.
\end{align}
Given a set of $N$ measurements $\{y^{(1)},\ldots,y^{(N)}\}$ and $\{x^{(1)},\ldots,x^{(N)} \}$, where $y^{(i)}=f(x^{(i)})+w^{(i)}$, $i\in\{1,\ldots,N\}$ as in Assumption \ref{A3}, the posterior distribution corresponding to $f_j(x)$, for $j\in\{1,\ldots,n\}$ at an arbitrary state $x\in X$ is  computed as a normal distribution $\mathcal N(\mu_j(x),\rho_j(x))$ with mean and covariance
\begin{align}
\mu_j(x)&=\overline k_j^T( K_j+\rho_f^2\bm I_N)^{-1}y_j,\label{mean}\\
\rho_j^2(x)&=k_j(x,x)-\overline k_j^T( K_j+\rho_f^2\bm I_N)^{-1}\overline k_j,\label{SD}
\end{align}
where $\overline k_j=[k_j(x^{(1)},x)\cdots k_j(x^{(N)},x)]^T\in\R^N$, $y_j=[y_j^{(1)} \cdots y_j^{(N)}]^T\in\R^N$, and 
\begin{align*}
K_j=\begin{bmatrix}
k_j(x^{(1)},x^{(1)})  &\cdots   & k_j(x^{(1)},x^{(N)}) \\ 
\vdots & \ddots  & \vdots \\ 
k_j(x^{(N)},x^{(1)}) &  \cdots & k_j(x^{(N)},x^{(N)}) 
\end{bmatrix}\in\R^{N\times N}. 
\end{align*}
Now consider the bound $\overline{\rho}_j^2=\max_{x\in X}\rho_j^2(x)$. 
The existence of such bound follows from continuity of kernels. The approximation of overall $f$ can be obtained by concatenating $\mu_j$ and $\rho_j$ in \eqref{mean} and \eqref{SD} as follows    
\begin{align}
\mu(x)&:=[\mu_1(x), \ldots, \mu_n(x)]^T,\label{eq:mu_gp}\\
\rho^2(x)&:=[\rho^2_1(x), \ldots, \rho^2_n(x)]^T.\label{eq:rho_gp}
\end{align}
Considering Assumption \ref{A2}, one can upper bound the difference between true value of $f_j(x)$  and inferred mean $\mu_j(x)$ with high probability as given in the next proposition.
\begin{proposition}\label{lemma1}
	Consider a system $\mathcal{S}$ in \eqref{sys} with Assumptions \ref{A2} and \ref{A3}, and the learned Gaussian process model with mean $\mu_j$ and standard deviation $\rho_j$ as given in \eqref{mean} and \eqref{SD}, respectively. Then, the model error is bounded by
	\begin{align}\label{probb1}
	\mathbb{P}\Big\{&\hspace{-.2em}\bigcap_{j=1}^n\hspace{-.2em}\mu_j(x)\hspace{-.2em}-\hspace{-.2em}\beta_j\rho_j(x)\hspace{-.1em}\leq\hspace{-.1em} f_j(x) \hspace{-.1em}\leq\hspace{-.1em}\mu_j(x)\hspace{-.2em}+\hspace{-.2em}\beta_j\rho_j(x),\forall x\in X\Big\}\nonumber\\&\geq(1-\varepsilon)^n,
	\end{align}
	with $\varepsilon\in(0,1)$ and $\beta_j:=\sqrt{2\|f_j\|_{k_j}^2+300\gamma_j \log^3(\frac{N+1}{\varepsilon})}$, where $N$ is the number of finite data samples and  
	$\gamma_j$ is the information gain $($see Remark \ref{info_gain}$)$.
\end{proposition} 
\begin{proof}
	The proof is similar to that of \cite[Lemma 2]{umlauft2018uncertainty} and follows from \cite[Theorem 6]{srinivas2012information} by extending the inequality (given for the scalar case) $\mathbb{P}\Big\{\mu_j(x)-\beta_j\rho_j(x)\leq f_j(x) \leq\mu_j(x)+\beta_j\rho_j(x),\forall x\in X\Big\}\geq(1-\varepsilon)$ to $n$-dimensional state-set.
\end{proof}
By using the bound $\overline\rho_j$, one can rewrite \eqref{probb1} as
\begin{align}\label{aaa}
\mathbb{P}\Big\{f(x)\in\{\mu(x) +d\mid d\in\mathcal{D}\}, \forall x\in X\Big\}\geq(1-\varepsilon)^n,
\end{align}
where $\mathcal{D}:=\{[d_1,\ldots, d_n]^T\mid d_j\in[-\beta_j\overline{\rho}_j, \beta_j\overline{\rho}_j], j\in\{1,\ldots,n\}\}$.
\begin{remark}\label{info_gain}
	The information gain $\gamma_j$ in Proposition \ref{lemma1} quantifies the maximum mutual information between a finite set of data-samples and actual function $f_j$. Exact evaluation of $\gamma_j$ is an NP-hard problem in general, however, it can be greedily approximated and has a sublinear dependency on the number of data-samples $N$ for many commonly used kernels. For the detailed discussion regarding the upper bound on $\gamma_j$, we kindly refer the interested readers to the work in \cite{srinivas2012information}.
\end{remark}
\section{Control Barrier Functions}\label{sec:cbf}
In this section, we provide sufficient conditions using so-called control barrier functions under which, we can show the result providing guarantees on safety specifications. First, we provide the result assuming the availability of full knowledge of the system. Then, we use it to provide the result for unknown systems. The following theorem is inspired by the result of Proposition 2 in \cite{prajna2007framework}.       
\begin{theorem}\label{lemma12}
	Consider a system $\mathcal{S}$ in \eqref{sys} and sets $X_0$, $X_1\subseteq X$. Suppose there exists a differentiable function $B:X\ra\R$ satisfying following conditions
	\begin{align}
	\forall x\in X_0\qquad\quad& B(x)\leq 0, 	 \label{ineq1}\\
	\forall x\in X_1 \qquad\quad& B(x)> 0, \qquad\quad	\label{ineq2}\\
	\forall x\in X \ \exists u\in U \qquad\quad&\frac{\partial B}{\partial  x}(x)(f(x)+g(x)u)\leq 0.\label{ineq3}
	\end{align}
	Then, for a trajectory $\phi$ of system \eqref{sys} starting from any $x_0\in X_0$ under a control policy $\upsilon$ associated to $B$ $($cf. condition \eqref{ineq3}$)$, one gets $\phi(x_0,\upsilon,t)\notin X_1$ $\forall t\in\R^+_0$.
\end{theorem}
\begin{proof}
	We prove by contradiction. Suppose that there exists $t\in\R_0^+$ such that the trajectory of $\mathcal{S}$ starting from $x_0 \in X_0$, $\phi(x_0,\upsilon,t)$ reaches a state inside $X_1$. Following \eqref{ineq1} and \eqref{ineq2}, one has $B(\phi(x_0,\upsilon,0))\leq 0$ and $B(\phi(x_0,\upsilon,t))>0$ for some $t\in\R_0^+$. By using inequality \eqref{ineq3}, one can conclude that $B(\phi(x_0,\upsilon,t))\leq B(\phi(x_0,\upsilon,0))\leq 0$ for all $t\in\R^+_0$. This contradicts our assumption and concludes the proof.
\end{proof}	
The function $B$ in Theorem \ref{lemma12} satisfying \eqref{ineq1}-\eqref{ineq3} is usually referred to as a control barrier function.
\begin{remark}
	Condition \eqref{ineq3} implicitly associates a controller $\mathbf{u}:X\ra U$ according to the existential quantifier on $u$ for any $x\in X$. The stationary control policy $\upsilon$ driving the system is readily given by $\upsilon(t)=\mathbf{u}(\phi(t))$, where $\phi(t)$ is the value of the trajectory of the system at time instance $t\in\R_0^+$. 
\end{remark}
Now, in order to extend the above result to unknown systems described in Section \ref{prelim}, we first learn the Gaussian process model as discussed in Section \ref{gpm}.

Particularly, the next result provides a probabilistic guarantee on the result given in Theorem \ref{lemma12}.
{\begin{theorem}\label{thm1}
		Consider a system $\mathcal{S}$ in \eqref{sys}, the learned Gaussian process model with mean $\mu(\cdot)$ and covariance $\rho^2(\cdot)$ as given in \eqref{eq:mu_gp} and \eqref{eq:rho_gp}, and the result in Proposition \ref{lemma1}. Let $X_0$, $X_1\subseteq X$. Suppose there exists a differentiable function $B:X\ra\R$ such that
		\begin{align}
		\forall x\in X_0\quad& B(x)\leq 0, 	 \label{ineq11}\\
		\forall x\in X_1 \quad& B(x)> 0,	\label{ineq21}\\
		\forall x\in X \ \exists u\in U \ \forall  d\in\mathcal{D}\quad&\frac{\partial B}{\partial x}(x)( \mu(x)+ d+g(x) u)\leq 0  \label{ineq31},
		\end{align}
		where $\mathcal{D}$ is the set defined in \eqref{aaa}. Then, for a trajectory $\phi$ of \eqref{sys} starting from any $x_0\in X_0$ under a control policy $\upsilon$ associated to $B$ $($cf. condition \eqref{ineq31}$)$, one has $\phi(x_0,\upsilon,t)\not\in X_1$ for all $t\in\R^+_0$ with probability at least $(1-\varepsilon)^n$.
	\end{theorem}
	\begin{proof}
		The first two inequalities are equivalent to \eqref{ineq1} and \eqref{ineq2}. Now consider condition \eqref{ineq31}.
		Following the result of Proposition \ref{lemma1}, we have $\mathbb{P}\Big\{f(x)\in\{\mu(x) +d\mid d\in\mathcal{D}\}, \forall x\in X\Big\}\geq(1-\varepsilon)^n$ which implies the following 
		$$\mathbb{P}\{\forall x\in X \ \exists u\in U,\ \frac{\partial B}{\partial x}(x)(f(x)+g(x)u)\leq 0\}\geq (1-\varepsilon)^n.$$
		Let us consider following events: $A_1$ representing existence of $B$ satisfying \eqref{ineq11}-\eqref{ineq31}, $A_2$ representing existence of $B$ satisfying \eqref{ineq1}-\eqref{ineq3}, and $A_3$ representing that $\phi(x_0,\upsilon,t)\notin X_1$ $\forall x_0\in X_0$ $\forall t\in\R^+_0$.	
		Now,
		$\mathbb{P}\{A_1\implies A_3\}\geq\mathbb{P}\{(A_1\implies A_2) \text{ and } (A_2\implies A_3 )\}=\mathbb{P}\{A_1\implies A_2\} \mathbb{P}\{A_2\implies A_3\}\geq (1-\varepsilon)^n$. The last equality follows from the fact that the events $A_1\implies A_2$ and $A_2\implies A_3$ are mutually independent. The last inequality is obtained by using $\mathbb{P}\{A_2\implies A_3\}=1$ which follows from Theorem \ref{lemma12} and $\mathbb{P}\{A_1\implies A_2\}\geq (1-\varepsilon)^n$ showed in the first part of the proof. Thus, the existence of $B$ satisfying \eqref{ineq11}-\eqref{ineq31} implies that
		for any $x_0\in X_0$, one has $\phi(x_0,\upsilon,t)\notin X_1$ for all $t\in\mathbb{R}^+_0$ with probability at least $(1-\varepsilon)^n$. 
\end{proof}}
Note that condition \eqref{ineq31} in Theorem \ref{thm1} associates a control policy $\upsilon$ according to the existential quantifier on $u$ which provides input trajectory $\upsilon$ given state trajectory $\phi$ and for any arbitrary choice of $d\in\mathcal{D}$.
\begin{remark}
    One can easily extend the result approach to provide a probabilistic guarantee on the satisfaction of more general temporal specifications expressed using infinite strings over deterministic co-Buchi automata (DCA) by decomposing it to a sequence of safety specifications as discussed in \cite[Section 4.1]{jagtap2020compositional}. Then, we can combine the probability bounds and corresponding controllers obtained with the help of Theorem \ref{thm1} for each safety property to obtain a hybrid control policy that provides probability bound on the satisfaction of specification given by DCA (see \cite[Section IV]{jahanshahi2020synthesis} for a similar result). 
\end{remark}
\subsection{Computation of Control Barrier Functions}\label{cegis}
Proving the existence of a control barrier function and finding one are in general hard problems. 
However, under the assumption of having a finite input set $U=\{u_1,u_2,\ldots,u_l\}$, where $u_i\in\R^m$, $i\in\{1,2,\ldots,l\}$, one can search for a parametric control barrier functions and corresponding control policies.
To search for control barrier functions, we propose the use of counterexample guided inductive synthesis (CEGIS) approach. 
This approach uses feasibility solvers for finding control barrier functions of a given parametric form using Satisfiability Modulo Theories (SMT) solvers such as \texttt{Z3}  \cite{z3}, \texttt{MathSAT}  \cite{cimatti2013mathsat5}, or \texttt{dReal}  \cite{Gao.2013}. 
The feasibility condition that is required to be satisfied for the existence of a control barrier function $B$ is given in the following lemma.
\begin{lemma} \label{cbclem}
	Consider a control affine system $\mathcal S $ with a finite input set $U=\{u_1,u_2,\ldots,u_l\}$, where $u_i\in\R^m$, $i\in\{1,2,\ldots,l\}$. Suppose there exists a function $B(x)$ such that the following expression is true
	\begin{align}\label{feas1}
	&\bigwedge_{x \in X_0}B(x) \leq 0 \bigwedge_{x \in X_{1}} B(x) > 0\nonumber\\
	&\bigwedge_{x \in X}\hspace{-.2em}\Big(\bigvee_{u\in U}\hspace{-.2em}\Big(\bigwedge_{d \in \mathcal{D}}\hspace{-.2em}\big(\frac{\partial B}{\partial x}(x)(\mu(x)+d+g(x)u)\hspace{-.2em}\leq\hspace{-.2em} 0\big)\Big)\Big).
	\end{align}
	Then, $B(x)$ satisfies conditions \eqref{ineq11}-\eqref{ineq31} in Theorem \ref{thm1} and any $\mathbf{u}:X\rightarrow U$ with $\mathbf{u}(x):=\{u\in U\mid\frac{\partial B}{\partial x}( \mu(x)+ d+g(x)u)\leq 0\}$ for any arbitrary choice of $d\in\mathcal D$ is a corresponding controller.
\end{lemma}
In order to find a function $B$ in Lemma \ref{cbclem} using CEGIS framework, we consider a function of the parametric form $B(a,x)= \sum_{i=1}^{p}a_i b_i(x)$ with some user-defined (nonlinear) basis functions $b_i(x)$ and unknown coefficients $a_i \in \mathbb{R}, i \in \{1, 2,\ldots, p\}$. With this choice of barrier function the feasibility expression \eqref{feas1} can be rewritten as
\begin{align}\label{CEGIS_algo}
&\psi(a,x)\hspace{-0.2em}:=\bigwedge_{x \in X_0}B(a, x) \leq 0 \bigwedge_{x \in X_{1}} B(a, x) > 0   \nonumber \\ 
&\bigwedge_{x \in X}\hspace{-.2em}\Big(\hspace{-.2em}\bigvee_{u\in U}\hspace{-.2em}\Big(\hspace{-.2em}\bigwedge_{d \in \mathcal{D}}\hspace{-.2em}\big(\frac{\partial B}{\partial x}(a,x)(\mu(x)+d+g(x)u)\hspace{-.2em}\leq \hspace{-.2em}0\big)\hspace{-.1em}\Big)\hspace{-.1em}\Big).
\end{align}
The coefficients $a_i$ can be efficiently found using SMT solvers such as \texttt{Z3} or \texttt{MathSAT} for the finite set $\overline{X}\subset X$ of data samples. We denote the obtained candidate barrier function with fixed coefficients $a_i$ by $B(a,x)|_a$ and the corresponding feasibility expression by $\psi(a,x)|_a$. Next we obtain counterexample $x\in X$ such that $B(a,x)|_a$ satisfies $\neg \psi(a,x)|_a$. If $\neg \psi(a,x)|_a$ has no feasible solution, then the obtained $B(a,x)|_a$ is a true control barrier function. If $\neg \psi(a,x)|_a$ is feasible, we update data samples as $\overline{X}=\overline{X}\cup x$ and recompute coefficients $a_i$ iteratively until $\neg \psi(a,x)|_a$ becomes infeasible. For the detailed discussion on the CEGIS approach, we kindly refer interested readers to \cite[Subsection 5.3.2]{jagtap2019formal}. The pseudocode for CEGIS framework to compute such control barrier functions is given in Algorithm \ref{algo}.
Note that, this CEGIS algorithm either $(i)$ runs forever, or $(ii)$ terminates after some finite iterations with a control barrier function satisfying \eqref{feas1}, or $(iii)$ terminates with a counter example proving that no solution exists for the considered structure of barrier function. In order to guarantee termination of the algorithm, one can fix some upper bound on the number of unsuccessful iterations.\\
\begin{algorithm}[t] 
	\caption{CEGIS Framework}
	\label{algo}
	\begin{algorithmic}[1]
		\State Define $\overline{X}\subset X$\Comment{set of finite data samples in $X$} 
		\State Define $B(a,x):= \sum_{i=1}^{p}a_ib_i(x)$
		\While {True}
		\If{$\psi(a,x)$ is \textit{unsat}} 
		\State \textit{infeasible} 
		\State \textbf{break}
		\Else
		\State Compute candidate $B(a,x)|_a$ satisfying \eqref{CEGIS_algo} \State for all $x\in\overline{X}$
		\If{$\neg \psi(a,x)|_a$ is \textit{unsat}}
		\State {$B(a,x)|_a$ is a control barrier function} 
		\State \textbf{break}
		\Else
		\State $cex=x\in X \ s.t.\  \neg \psi(a,x)|_a\text{ is }sat$
		\State $\overline{X} \leftarrow \overline{X}\cup cex$
		\EndIf
		\EndIf
		\EndWhile
	\end{algorithmic}
\end{algorithm}
 \section{Case Study}\label{sec:case}
 \begin{figure}[t] 
 	\centering
 	\subfigure[]{\includegraphics[scale=0.4]{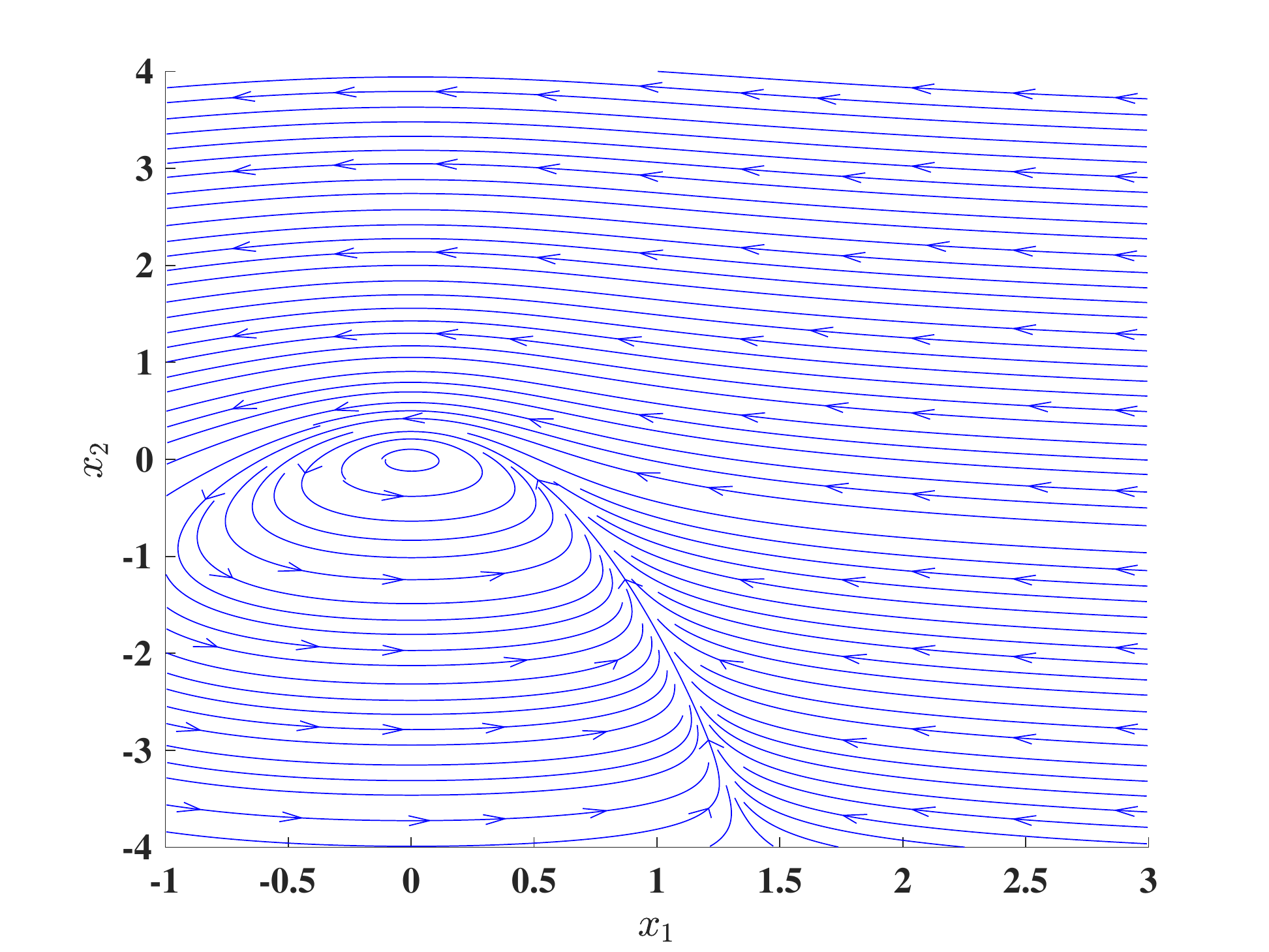}}\\\vspace{-1em}\subfigure[]{\includegraphics[scale=0.4]{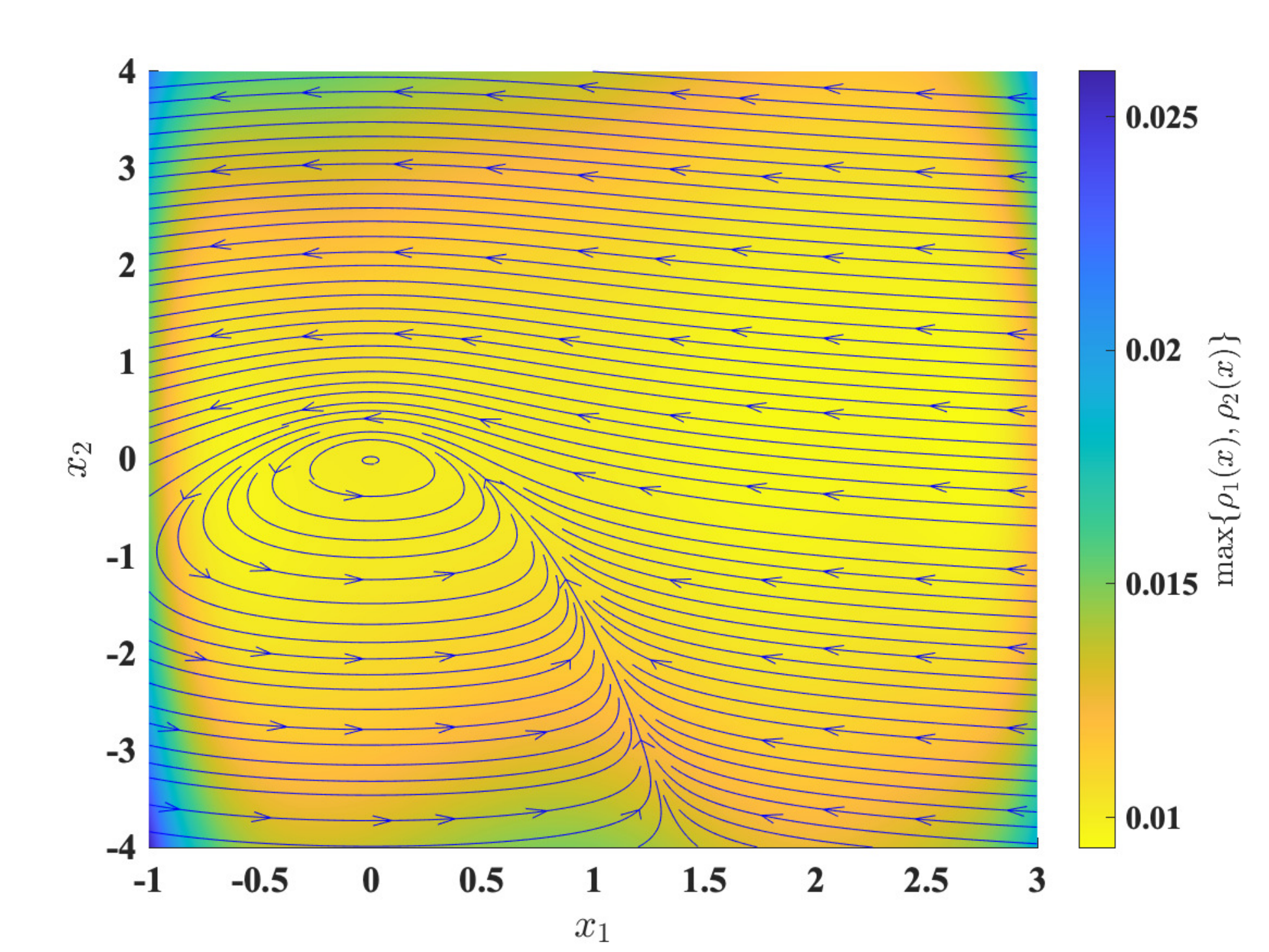}} 
 	\caption{Illustration of the vector field learned using GPs: (a) vector field $f(x)$ of the original system, (b) learned vector field ${\mu}(x)$ with the maximum standard deviation shown using colormap.}
 	\label{fig:learned_model}
 \end{figure}
We consider the nonlinear Moore-Greitzer jet engine model in no-stall mode used by \cite{zamani2011symbolic} as a case study.
Consider the unknown nonlinear dynamics given as  
\begin{align*}
f(x)= \begin{bmatrix}
f_1(x)\\ 
f_2(x)
\end{bmatrix}=\begin{bmatrix}
-x_2-\frac{3}{2}x_1^2-\frac{1}{2}x_1^3\\ 
x_1
\end{bmatrix} \text{ and } g(x)=\begin{bmatrix}
0\\ 
-1
\end{bmatrix},
\end{align*}
where $x=[x_1,x_2]^T$, $x_1=\Phi-1$, $x_2=\Psi-\psi-2$, $\Phi$ is the mass flow, $\Psi$ is the pressure rise, and $\psi$ is a constant. The control input $u\in U=\{-2,-1.5,-1,-0.5,0,0.5,1,1.5,2\}$. We consider a compact state-space $X=[-1,3]\times[-4,4]$ and $X_0=[0,1]\times[-1,1]$, $X_1=[-1,0]\times[-4,-2.5] \cup [-1,3]\times[2,4]$. 
Note that functions $f_1$ and $f_2$ are continuous thus they satisfy Assumption \ref{A2}  \cite{srinivas2012information}. 

In order to synthesize a control policy ensuring safety specification, we first learn the unknown model using Gaussian processes. For learning Gaussian process model, we collected $35$ data samples of $x$ and $y=f(x)+w$, where $w \sim \mathcal{N}(0,\rho_f^2\bm I_2)$, $\rho_f=0.01$, by simulating the system with several initial states chosen randomly using uniform distribution. We used squared-exponential kernel  \cite{williams2006gaussian} defined as 
$k_j(x,x')=\rho_{k_j}^2\exp\Big(\sum_{i=1}^2\frac{(x_i-x_i')^2}{-2l_{ji}^2}\Big), j\in\{1,2\}$, where $\rho_{k_1}^2=224.4168$ and $\rho_{k_2}^2=24.5311$ are signal variances and $l_{11}=6.6030$, $l_{12}=327.5503$, $l_{21}=42.1995$, and $l_{22}=6.4648\times 10^6$ are length scales. These parameters are obtained through likelihood maximization using a quasi-Newton method. The inferred mean and variance are represented as in \eqref{eq:mu_gp} and \eqref{eq:rho_gp} with $\overline{\rho}_{\max}=\max\{\overline{\rho}_1,\overline{\rho}_2\}=0.0360$. Figure \ref{fig:learned_model} illustrates the actual and the learned vector fields.
Computing $\|f_j\|_{k_j}$ and $\gamma_j$, $j\in\{1,2\}$, is a hard problem in general. Thus, we employ Monte-Carlo approach to obtain the probability bound on the accuracy of the learned model provided in Lemma \ref{lemma1}. 
For a fix error bound on the distance between actual vector field $f_j(x)$ and inferred mean $\mu_j(x)$ that is $\beta_j\overline\rho_j=0.05$, $j=1,2$ (i.e., $\mathcal{D}=[-0.05, 0.05]^2$), we obtained a probability interval for the probability in \eqref{aaa} as
$\mathbb{P}\Big\{f(x)\in\{\mu(x) + d\mid  d\in\mathcal{D}\}, \forall x\in X\Big\}\in[0.9839, 0.9905]$ with confidence $1-10^{-10}$ using $10^6$ realizations.
Thus, one can choose the lower bound $(1-\epsilon)^2$ as $0.9839$. Note that for a fix error bound, if we increase the number of data samples $N$ used for learning GPs, we get the tighter lower bound on the probability in \eqref{aaa}. Figure \ref{fig:relation} shows the effect of the increase in number of data samples $N$ on the lower bound on the probability (top) and the maximum standard deviation $\overline{\rho}_{\max}$ (bottom).  \\
\begin{figure}[t] 
	\centering
	\includegraphics[scale=0.42]{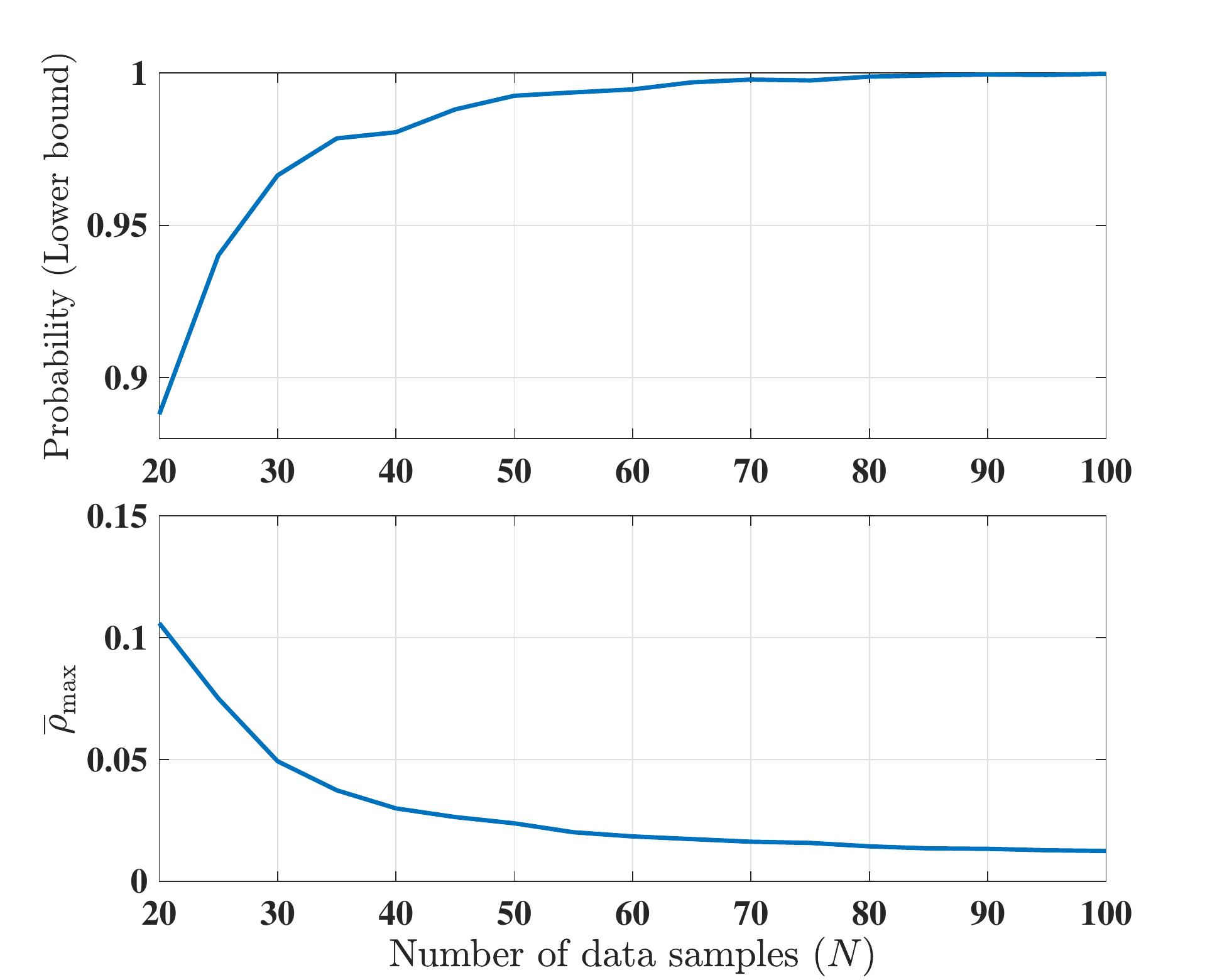} 
	\caption{The change in lower bound on the probability (top) and the maximum standard deviation $\overline{\rho}_{\max}$ (bottom) with increase in the number of data samples $N$.}
	\label{fig:relation}
\end{figure}
Next, we compute a polynomial type control barrier function of order 2  using CEGIS approach discussed in Subsection \ref{cegis} as the following: 
\begin{align*}
B(x)= &- 4292.8910 + 1129.2414x_1+ 1010.3266x_2 \\&+ 1274.3322x_1^2 + 1564.8195x_2^2 - 1368.6064x_1x_2.
\end{align*}
The corresponding controller is given by 
\begin{align}\label{controller123}
\mathbf u(x)\hspace{-.2em}=\hspace{-.2em}\min\{u\in U\mid\frac{\partial B}{\partial x}(x)( \mu(x)\hspace{-.1em}+ \hspace{-.1em}d\hspace{-.1em}+\hspace{-.1em}g(x)u)\leq 0\},
\end{align}
for an arbitrarily chosen $ d \in [-0.05, 0.05]^2$.
Using Theorem \ref{thm1}, one can have a lower bound on the probability that the trajectories of the system starting from any initial state $x_0\in X_0$ under the control policy $\upsilon$ corresponding to controller in \eqref{controller123} satisfy the safety specification as $0.9839$. 
The closed-loop trajectories using controller \eqref{controller123} starting from several initial conditions in $X_0$ are shown in Figure \ref{fig:resp}.  
\begin{figure}[t] 
	\centering
	\includegraphics[scale=0.4]{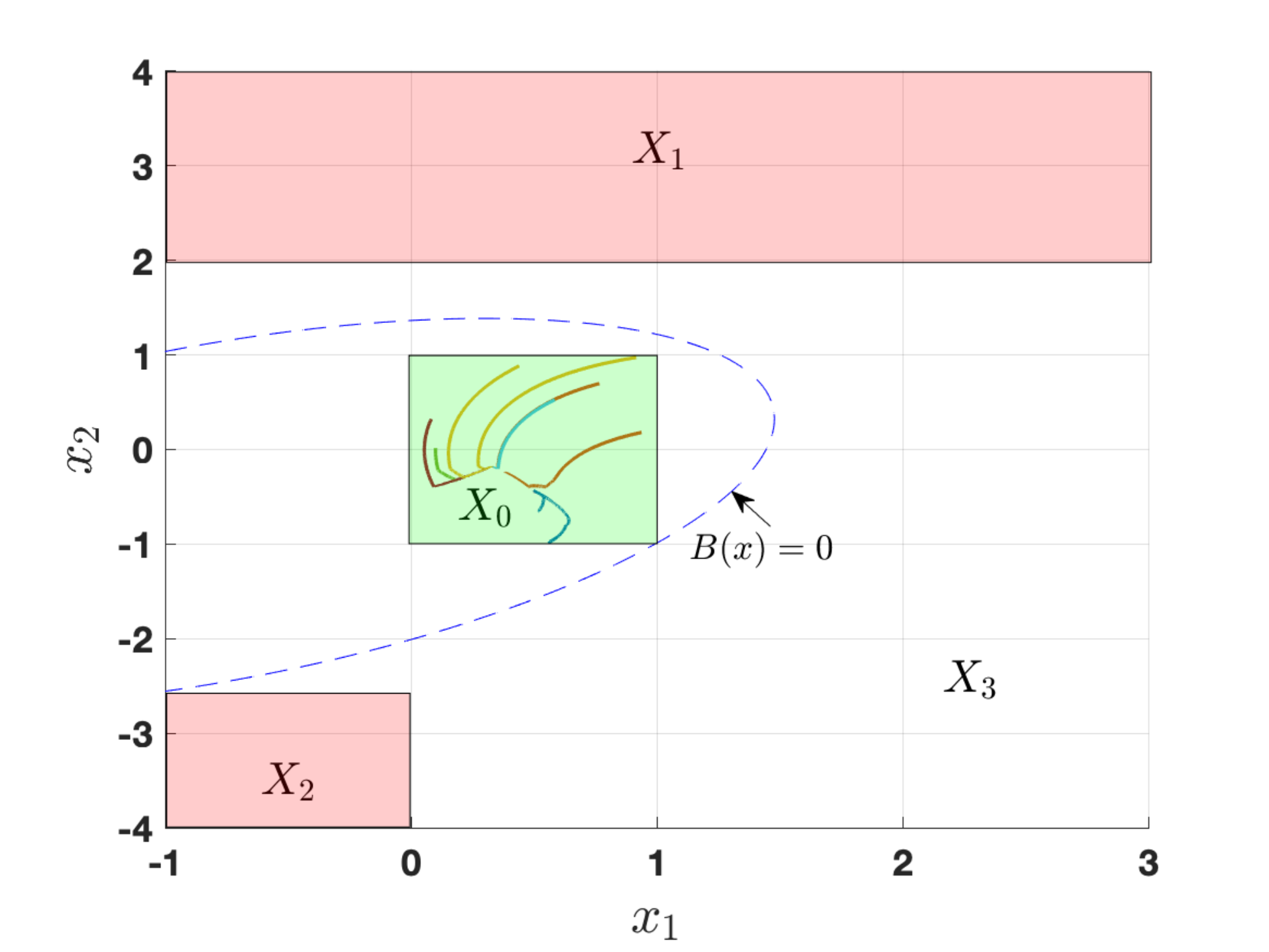} 
	\caption{The solid colored lines are closed-loop trajectories starting from several initial states in $X_0$. The dashed curve shows the zero level set of $B(x)$, defined as $\{x\in X\mid B(x)=0\}$ and the regions of interest are shown using colored rectangles.}
	\label{fig:resp}
\end{figure}
\section{CONCLUSIONS}\label{conclusion}
In this work, we proposed a scheme for synthesizing control policies for unknown control affine systems enforcing safety specifications. 
We provide a systematic technique to solve safety problems by computing control barrier functions utilizing models learned through Gaussian processes with high confidence.


\bibliographystyle{IEEEtran}
\bibliography{bibliography}

\begin{thebibliography}{10}
\providecommand{\url}[1]{#1}
\csname url@samestyle\endcsname
\providecommand{\newblock}{\relax}
\providecommand{\bibinfo}[2]{#2}
\providecommand{\BIBentrySTDinterwordspacing}{\spaceskip=0pt\relax}
\providecommand{\BIBentryALTinterwordstretchfactor}{4}
\providecommand{\BIBentryALTinterwordspacing}{\spaceskip=\fontdimen2\font plus
\BIBentryALTinterwordstretchfactor\fontdimen3\font minus
  \fontdimen4\font\relax}
\providecommand{\BIBforeignlanguage}[2]{{%
\expandafter\ifx\csname l@#1\endcsname\relax
\typeout{** WARNING: IEEEtran.bst: No hyphenation pattern has been}%
\typeout{** loaded for the language `#1'. Using the pattern for}%
\typeout{** the default language instead.}%
\else
\language=\csname l@#1\endcsname
\fi
#2}}
\providecommand{\BIBdecl}{\relax}
\BIBdecl

\bibitem{Tabuada.2009}
P.~Tabuada, \emph{Verification and control of hybrid systems: a symbolic
  approach}.\hskip 1em plus 0.5em minus 0.4em\relax Springer Science \&
  Business Media, 2009.

\bibitem{belta2017formal}
C.~Belta, B.~Yordanov, and E.~A. Gol, \emph{Formal methods for discrete-time
  dynamical systems}.\hskip 1em plus 0.5em minus 0.4em\relax Springer, 2017,
  vol.~89.

\bibitem{kocijan2016modelling}
J.~Kocijan, \emph{Modelling and control of dynamic systems using Gaussian
  process models}.\hskip 1em plus 0.5em minus 0.4em\relax Springer, 2016.

\bibitem{kocijan2004gaussian}
J.~Kocijan, R.~Murray-Smith, C.~E. Rasmussen, and A.~Girard, ``Gaussian process
  model based predictive control,'' in \emph{Proceedings of the 2004 American
  control conference}, vol.~3.\hskip 1em plus 0.5em minus 0.4em\relax IEEE,
  2004, pp. 2214--2219.

\bibitem{koller2018learning}
T.~Koller, F.~Berkenkamp, M.~Turchetta, and A.~Krause, ``Learning-based model
  predictive control for safe exploration,'' in \emph{2018 IEEE Conference on
  Decision and Control (CDC)}.\hskip 1em plus 0.5em minus 0.4em\relax IEEE,
  2018, pp. 6059--6066.

\bibitem{hewing2019cautious}
L.~Hewing, J.~Kabzan, and M.~N. Zeilinger, ``Cautious model predictive control
  using gaussian process regression,'' \emph{IEEE Transactions on Control
  Systems Technology}, 2019.

\bibitem{chowdhary2014bayesian}
G.~Chowdhary, H.~A. Kingravi, J.~P. How, and P.~A. Vela, ``Bayesian
  nonparametric adaptive control using gaussian processes,'' \emph{IEEE
  transactions on neural networks and learning systems}, vol.~26, no.~3, pp.
  537--550, 2014.

\bibitem{beckers2019stable}
T.~Beckers, D.~Kuli{\'c}, and S.~Hirche, ``Stable gaussian process based
  tracking control of euler--lagrange systems,'' \emph{Automatica}, vol. 103,
  pp. 390--397, 2019.

\bibitem{capone2019backstepping}
A.~Capone and S.~Hirche, ``Backstepping for partially unknown nonlinear systems
  using gaussian processes,'' \emph{IEEE Control Systems Letters}, vol.~3,
  no.~2, pp. 416--421, 2019.

\bibitem{8368325}
J.~{Umlauft}, L.~{P{\"o}hler}, and S.~{Hirche}, ``An uncertainty-based control
  lyapunov approach for control-affine systems modeled by gaussian process,''
  \emph{IEEE Control Systems Letters}, vol.~2, no.~3, pp. 483--488, July 2018.

\bibitem{umlauft2017feedback}
J.~Umlauft, T.~Beckers, M.~Kimmel, and S.~Hirche, ``Feedback linearization
  using gaussian processes,'' in \emph{2017 IEEE 56th Annual Conference on
  Decision and Control (CDC)}.\hskip 1em plus 0.5em minus 0.4em\relax IEEE,
  2017, pp. 5249--5255.

\bibitem{berkenkamp2016safe}
F.~Berkenkamp, A.~P. Schoellig, and A.~Krause, ``Safe controller optimization
  for quadrotors with gaussian processes,'' in \emph{2016 IEEE International
  Conference on Robotics and Automation (ICRA)}.\hskip 1em plus 0.5em minus
  0.4em\relax IEEE, 2016, pp. 491--496.

\bibitem{akametalu2014reachability}
A.~K. Akametalu, J.~F. Fisac, J.~H. Gillula, S.~Kaynama, M.~N. Zeilinger, and
  C.~J. Tomlin, ``Reachability-based safe learning with gaussian processes,''
  in \emph{53rd IEEE Conference on Decision and Control}.\hskip 1em plus 0.5em
  minus 0.4em\relax IEEE, 2014, pp. 1424--1431.

\bibitem{ames2016control}
A.~D. Ames, X.~Xu, J.~W. Grizzle, and P.~Tabuada, ``Control barrier function
  based quadratic programs for safety critical systems,'' \emph{IEEE
  Transactions on Automatic Control}, vol.~62, no.~8, pp. 3861--3876, 2016.

\bibitem{ames2019control}
A.~D. {Ames}, S.~{Coogan}, M.~{Egerstedt}, G.~{Notomista}, K.~{Sreenath}, and
  P.~{Tabuada}, ``Control barrier functions: Theory and applications,'' in
  \emph{18th European Control Conference (ECC)}, June 2019, pp. 3420--3431.

\bibitem{jagtap2019formal}
P.~Jagtap, S.~Soudjani, and M.~Zamani, ``Formal synthesis of stochastic systems
  via control barrier certificates,'' \emph{IEEE Transactions on Automatic
  Control}, 2020.

\bibitem{yang2019continuous}
G.~Yang, C.~Belta, and R.~Tron, ``Continuous-time signal temporal logic
  planning with control barrier functions,'' in \emph{2020 American Control
  Conference (ACC)}.\hskip 1em plus 0.5em minus 0.4em\relax IEEE, 2020, pp.
  4612--4618.

\bibitem{li2019temporal}
X.~Li and C.~Belta, ``Temporal logic guided safe reinforcement learning using
  control barrier functions,'' \emph{arXiv preprint arXiv:1903.09885}, 2019.

\bibitem{wang2018safe}
L.~Wang, E.~A. Theodorou, and M.~Egerstedt, ``Safe learning of quadrotor
  dynamics using barrier certificates,'' in \emph{2018 IEEE International
  Conference on Robotics and Automation (ICRA)}.\hskip 1em plus 0.5em minus
  0.4em\relax IEEE, 2018, pp. 2460--2465.

\bibitem{cheng2019end}
R.~Cheng, G.~Orosz, R.~M. Murray, and J.~W. Burdick, ``End-to-end safe
  reinforcement learning through barrier functions for safety-critical
  continuous control tasks,'' in \emph{Proceedings of the AAAI Conference on
  Artificial Intelligence}, vol.~33, 2019, pp. 3387--3395.

\bibitem{paulsen2016introduction}
V.~I. Paulsen and M.~Raghupathi, \emph{An introduction to the theory of
  reproducing kernel Hilbert spaces}.\hskip 1em plus 0.5em minus 0.4em\relax
  Cambridge University Press, 2016, vol. 152.

\bibitem{seeger2008information}
M.~W. Seeger, S.~M. Kakade, and D.~P. Foster, ``Information consistency of
  nonparametric gaussian process methods,'' \emph{IEEE Transactions on
  Information Theory}, vol.~54, no.~5, pp. 2376--2382, 2008.

\bibitem{williams2006gaussian}
C.~K. Williams and C.~E. Rasmussen, \emph{Gaussian processes for machine
  learning}.\hskip 1em plus 0.5em minus 0.4em\relax MIT press Cambridge, MA,
  2006, vol.~2, no.~3.

\bibitem{umlauft2018uncertainty}
J.~Umlauft, L.~P{\"o}hler, and S.~Hirche, ``An uncertainty-based control
  {Lyapunov} approach for control-affine systems modeled by gaussian process,''
  \emph{IEEE Control Systems Letters}, vol.~2, no.~3, pp. 483--488, 2018.

\bibitem{srinivas2012information}
N.~Srinivas, A.~Krause, S.~M. Kakade, and M.~W. Seeger, ``Information-theoretic
  regret bounds for {Gaussian} process optimization in the bandit setting,''
  \emph{IEEE Transactions on Information Theory}, vol.~58, no.~5, pp.
  3250--3265, 2012.

\bibitem{prajna2007framework}
S.~Prajna, A.~Jadbabaie, and G.~J. Pappas, ``A framework for worst-case and
  stochastic safety verification using barrier certificates,'' \emph{IEEE
  Transactions on Automatic Control}, vol.~52, no.~8, pp. 1415--1428, 2007.

\bibitem{jagtap2020compositional}
P.~Jagtap, A.~Swikir, and M.~Zamani, ``Compositional construction of control
  barrier functions for interconnected control systems,'' in \emph{Proceedings
  of the 23rd International Conference on Hybrid Systems: Computation and
  Control}, 2020, pp. 1--11.

\bibitem{jahanshahi2020synthesis}
N.~Jahanshahi, P.~Jagtap, and M.~Zamani, ``Synthesis of partially observed
  jump-diffusion systems via control barrier functions,'' \emph{IEEE Control
  Systems Letters}, vol.~5, no.~1, pp. 253--258, 2020.

\bibitem{z3}
L.~de~Moura and N.~Bj{\o}rner, ``Z3: An efficient {SMT} solver,'' in
  \emph{Tools and algorithms for the construction and analysis of systems},
  ser. Lecture Notes in Computer Science, C.~R. Ramakrishnan and J.~Rehof,
  Eds.\hskip 1em plus 0.5em minus 0.4em\relax Berlin: Springer, 2008, vol.
  4963, pp. 337--340.

\bibitem{cimatti2013mathsat5}
A.~Cimatti, A.~Griggio, B.~J. Schaafsma, and R.~Sebastiani, ``The {MATHSAT5}
  {SMT} solver,'' in \emph{International Conference on Tools and Algorithms for
  the Construction and Analysis of Systems}.\hskip 1em plus 0.5em minus
  0.4em\relax Springer, 2013, pp. 93--107.

\bibitem{Gao.2013}
S.~Gao, S.~Kong, and E.~M. Clarke, ``d{R}eal: An {SMT} solver for nonlinear
  theories over the reals,'' in \emph{Automated deduction - CADE-24}, ser. LNCS
  sublibrary: SL 7 - artificial intelligence, M.~P. Bonacina, Ed.\hskip 1em
  plus 0.5em minus 0.4em\relax Heidelberg: Springer, 2013, vol. 7898, pp.
  208--214.

\bibitem{zamani2011symbolic}
M.~Zamani, G.~Pola, M.~Mazo, and P.~Tabuada, ``Symbolic models for nonlinear
  control systems without stability assumptions,'' \emph{IEEE Transactions on
  Automatic Control}, vol.~57, no.~7, pp. 1804--1809, 2011.

\end{thebibliography}

\addtolength{\textheight}{-12cm}   

\end{document}